\documentclass[aps,prl,reprint,nobalancelastpage]{revtex4-2}

\usepackage[utf8]{inputenc}

\usepackage{amsmath}
\usepackage{amssymb}
\usepackage[all]{xy}
\usepackage{tikz}
\usetikzlibrary{calc}

\usepackage[toc,page]{appendix}
\usepackage{enumerate}
\usepackage{tensor}
\usepackage{booktabs}
\usepackage{bbm}
\usepackage{youngtab}
\usepackage{slashed}
\usepackage{multirow}
\usepackage{makecell}
\usepackage{float}
\usepackage{comment}
\usepackage{verbatim}
\usepackage{xcolor}
\usepackage[framemethod=tikz]{mdframed}
\usepackage{soul}
\usepackage{mathtools}
\usepackage[inline]{enumitem}

\usepackage{amsthm}

\newtheorem{lemma}{Lemma}
\newtheorem{corollary}{Corollary}[lemma]
\newtheorem{remark}{Remark}[lemma]

\newcommand{\ab}{|}
\newcommand{\der}{\partial}
\newcommand{\de}{\mathrm{d}}

\newcommand{\sdil}{\phi}

\newcommand{\tsdil}{\tilde{\phi}}
\newcommand{\tomega}{\tilde{\omega}}
\newcommand{\tdelta}{\tilde{\delta}}
\newcommand{\tsigma}{\tilde{\sigma}}
\newcommand{\tvarphi}{\tilde{\varphi}}
\newcommand{\ttau}{\tilde{\tau}}

\newcommand{\e}{\mathrm{e}}

\newcommand{\p}{\mathrm{P}}

\usepackage{hyperref}
\hypersetup{
    colorlinks=true,
    linkcolor=purple,
    citecolor=purple,
    filecolor=purple,      
    urlcolor=purple,
}

\begin{document}
\numberwithin{equation}{section}

\title{Accelerating universe at the end of time}
\author{Gary Shiu}
\email{shiu@physics.wisc.edu}
\author{Flavio Tonioni}
\email{tonioni@wisc.edu}
\affiliation{Department of Physics, University of Wisconsin-Madison, 1150 University Avenue, Madison, WI 53706, USA}
\author{Hung V. Tran}
\email{hung@math.wisc.edu}
\affiliation{Department of Mathematics, University of Wisconsin-Madison, 480
Lincoln Drive, Madison, WI 53706, USA}

\begin{abstract}
We investigate whether an accelerating universe can be realized as an asymptotic late-time solution of FLRW-cosmology with multi-field multi-exponential potentials. Late-time cosmological solutions exhibit a universal behavior which enables us to bound the rate of time variation of the Hubble parameter. In string-theoretic realizations, if the dilaton remains a rolling field, our bound singles out a tension in achieving asymptotic late-time cosmic acceleration. Our findings go beyond previous no-go theorems in that they apply to arbitrary multi-exponential potentials and make no specific reference to vacuum or slow-roll solutions. We also show that if the late-time solution approaches a critical point of the dynamical system governing the cosmological evolution, the criterion for cosmic acceleration can be generally stated in terms of a directional derivative of the potential.
\end{abstract}

\maketitle

\section{Introduction}

The discovery of dark energy presents a deep challenge for quantum gravity. While a number of sophisticated scenarios for realizing de Sitter vacua in string theory have been developed \cite{Flauger:2022hie}, a fully explicit construction remains elusive. The root of the challenge is that the source of cosmic acceleration should be derived (rather than postulated) in a fundamental theory of gravity. It is a formidable task to demonstrate that the microphysics which stabilizes all the moduli would lead to a theoretically-controlled metastable de Sitter vacuum.

The Dine-Seiberg problem highlights the difficulty in finding parametrically weakly-coupled vacua \cite{Dine:1985he}. To avoid runaways to asymptotic regions of the moduli space, different-order terms in the moduli potential should necessarily compete. Having arbitrarily weak coupling would mean that there exist infinitely many vacua, or hidden parameters not related to vacuum expectation values of any field. This makes asymptotic runaway potentials an interesting alternative \cite{Obied:2018sgi, Ooguri:2018wrx}. Indeed, the observed small numbers and the approximate symmetries in nature suggest that the current universe may be approaching an asymptotic region of the field space. In this work, we study such asymptotic regions and prove a no-go theorem for an accelerating universe. As in many dynamical systems, the late-time regime exhibits some universal behaviors: this allows us to prove a bound on the rate of change of the Hubble parameter with only knowledge of the dimension of spacetime. The way we formulate this no-go statement also makes it clear how to evade it.

The main results of our paper are the following.
\begin{enumerate*}[label=(\roman*)]
    \item \label{result 1} We find a bound on the rate of time variation of the Hubble parameter at late time irrespective of whether stationary (vacua) or scaling solutions (which are the critical points of the dynamical system of interest) are reached.
    \item \label{result 2} This bound, when checked against string-theoretic constructions, imposes a generic obstacle to acceleration if the dilaton is one of the rolling fields. This also suggests ways out: for instance, if the dilaton is stabilized, or rolling in the non-asymptotic region, or if there are sufficiently many terms in the scalar potential (with terms of both signs necessarily present), the bound on acceleration is not automatically violated.
    \item \label{result 3} If a critical point is reached, we can express the proper measure of acceleration -- defined as the Hubble-parameter time variation -- in terms of a directional derivative, without assuming that a single term dominates in the potential or whether the kinetic or potential term dominates. We emphasize that in general the parameter $\smash{\epsilon = - \dot{H}/H^2}$, rather than the gradient of the potential commonly used as a swampland criterion, is the proper diagnostic for whether accelerating universes can occur.
\end{enumerate*}
The bound \ref{result 1} and the obstacle \ref{result 2} go beyond previous no-go results as we allow for quantum effects and we encompass vacua, non-vacua, slow-roll and non-slow-roll solutions. Detailed proofs are provided in appendix \ref{app: late-time bounds on cosmological autonomous systems}.

\section{Constraints on FLRW-cosmologies} \label{sec: FLRW-cosmology constraints}

The moduli space of string theory is typically curved. Nonetheless, string compactifications typically originate low-energy effective theories in which a number of canonically-normalized scalar fields $\phi^a$, for $a=1,\dots,n$, are subject to a scalar potential of the form
\begin{equation} \label{generic exponential potential}
    V = \sum_{i = 1}^m \Lambda_i \, \e^{- \kappa_d \gamma_{i a} \phi^a}.
\end{equation}
Here, $\Lambda_i$ and $\gamma_{ia}$ are constants that depend on the microscopic origin of the scalar potential, while $\kappa_d$ is the $d$-dimensional gravitational coupling. The set of scalars $\phi^a$ includes minimally the $d$-dimensional dilaton $\smash{\tdelta}$ and a radion $\smash{\tsigma}$ that controls the string-frame volume, unless these fields are stabilized at high energy scales. All theories with asymptotically-flat moduli spaces and without flat directions are amenable to our analysis. This is not a restrictive scenario since axions can be stabilized non-perturbatively. This general class of potentials also represents noteworthy phenomenological models by themselves and e.g. subsumes generalized assisted inflation \cite{Liddle:1998jc, Copeland:1999cs}. Let the non-compact $d$-dimensional spacetime be characterized by the usual FLRW-metric
\begin{equation*}
    d \tilde{s}_{1,d-1}^2 = - \de t^2 + a^2(t) \, \de l_{\mathbb{E}^{d-1}}^2,
\end{equation*}
with the Hubble parameter $\smash{H = \dot{a}/a}$ and $d>2$. One can reformulate the scalar-field and Friedmann equations in terms of an autonomous system of ordinary differential equations.

An accelerated cosmological expansion can only be achieved if the total scalar potential is positive: therefore, we focus on the scenario in which, at least asymptotically, $V>0$. Let $\smash{\Lambda_{i_+} > 0}$ and $\smash{\Lambda_{i_-} < 0}$ denote the positive- and negative-definite scalar-potential coefficients, respectively, distinguishing by the indices $\smash{i = i_+, i_-}$. For each  field $\phi^a$, let $\smash{\gamma_\pm^a = \min_{i_\pm} {\gamma_{i_\pm}}^a}$ and $\smash{\Gamma_\pm^a = \max_{i_\pm} {\gamma_{i_\pm}}^a}$. For all fields $\smash{\phi^a}$, let $\smash{\gamma_+^a \geq \Gamma_-^a}$: if $\smash{\gamma_+^a > 0}$, let $\smash{\gamma_\infty^a = \gamma_+^a}$; else, let $\smash{\gamma_\infty^a = 0}$. Then, if $\smash{(\gamma_\infty)^2 \leq 4 \, (d-1) / (d-2)}$, we are able to prove that, at all times $t > t_\infty$, where $t_\infty$ is a sufficiently large time, the $\epsilon$-parameter is bounded from below as
\begin{equation} \label{epsilon bound}
    \epsilon \geq \dfrac{d-2}{4} \, (\gamma_\infty)^2.
\end{equation}
Of course, the $\epsilon$-parameter is also bounded from above as $\smash{\epsilon \leq d - 1}$. If $\smash{\gamma_-^a \geq \Gamma_+^a}$ for a field, then one can redefine this field as $\smash{{\phi'}{}^a = -\phi^a}$ and find the same bound in terms of the flipped $\gamma$-coefficients. If for a field none of these orderings is in place, then there is no obvious contribution to the bound and thus we should set $\smash{\gamma_\infty^a = 0}$. If $\smash{(\gamma_\infty)^2 > 4 \, (d-1) / (d-2)}$, irrespective of the ordering of the $\smash{\gamma_\pm^a}$- and $\smash{\Gamma_\pm^a}$-coefficients, then the $\epsilon$-parameter asymptotically approaches the value $\epsilon = d-1$. All these statements are proven in appendix \ref{app: late-time bounds on cosmological autonomous systems}: see corollary \ref{corollary: f - upper bound} and remarks \ref{remark: bound hierarchy}-\ref{remark: general f - upper bound}, and lemma \ref{lemma: non-proper attractors} and remark \ref{remark: general non-proper attractors}. As physical observables are invariant under field-space $\mathrm{O}(n)$-rotations, the optimal version of the bound can be expressed as $\smash{\epsilon \geq [(d-2) / 4] \, \max_{\mathrm{R} \in \mathrm{O}(n)} [\gamma_\infty(\mathrm{R})]^2 = [(d-2) / 4] \, (\hat{\gamma}_\infty)^2}$, where $\smash{[\gamma_\infty(\mathrm{R})]^2}$ represents the $\smash{(\gamma_\infty)^2}$-coefficient in a rotated basis. In the optimal basis, the scalar potential reads $\smash{V = V_\infty + \sum_{\iota=m_\infty+1}^m \Lambda_\iota \, \e^{- \kappa_d \hat{\gamma}_{\iota \hat{\varphi}} \hat{\varphi} - \kappa_d \hat{\gamma}_{\iota \check{a}} \hat{\phi}^{\check{a}}}}$, where
\begin{align*}
    V_\infty = \biggl[ \sum_{\sigma=1}^{m_\infty} \Lambda_\sigma \, \e^{- \kappa_d \hat{\gamma}_{\sigma \check{a}} \hat{\phi}^{\check{a}}} \biggr] \, \e^{- \kappa_d \hat{\gamma}_\infty \hat{\varphi}}.
\end{align*}
Here, $\smash{\hat{\gamma}_\infty = \sqrt{(\hat{\gamma}_\infty)^2}}$ is the smallest-possible coupling for the field $\smash{\hat{\varphi}}$, with $\smash{\hat{\gamma}_\infty \leq \hat{\gamma}_{\iota \hat{\varphi}}}$, and $\smash{\hat{\phi}^{\check{a}}}$ are the further $n-1$ fields with their couplings $\smash{\hat{\gamma}_{i \check{a}}}$. In a single-field theory with the potential $\smash{\hat{V}_\infty = \hat{\Lambda}_\infty \, \e^{- \kappa_d \hat{\gamma}_\infty \hat{\varphi}}}$, the late-time $\epsilon$-parameter would be exactly $\smash{\epsilon = [(d-2)/4] \, (\hat{\gamma}_\infty)^2}$ if $\smash{(\hat{\gamma}_\infty)^2 \leq 4 \, (d-1) / (d-2)}$, or $\epsilon = d-1$ otherwise \cite{Copeland:1997et, Rudelius:2022gbz}. The presence of additional fields and couplings makes this just a lower bound.

A special situation is the one in which all terms in the potential are positive ($\smash{\Lambda_i > 0}$). In this case, there are no $\smash{\gamma_-^a}$- and $\smash{\Gamma_-^a}$-coefficients to compare with, and the bound in eq. (\ref{epsilon bound}) is automatically true. Also, our bound subsumes the special case of a single potential $\smash{V (\phi) = \Lambda \, \e^{- \kappa_d \gamma \phi}}$, where the late-time Hubble parameter takes the form $\smash{H = p / t}$, with $\smash{p = \mathrm{max} \, \lbrace 1/(d-1), 4 / \bigl[ (d-2) \gamma^2 \bigr] \rbrace}$ \cite{Copeland:1997et, Rudelius:2022gbz}; we further emphasize that it is generally not correct to assume that one exponential potential will dominate over the others, since for instance scaling solutions are such that all terms fall over time in exactly the same way. In string-theoretic constructions, the potentials generated by non-trivial curvature, NSNS-fluxes, heterotic Yang-Mills fluxes, type-II RR-fluxes, type-II D-brane/O-plane sources and generic Casimir-energy terms have the structure of eq. (\ref{generic exponential potential}). In fact, here $\smash{\tdelta}$ always appears with a $\gamma$-coefficient such that
\begin{equation} \label{dilaton-gamma lower bound}
    \gamma_{\tdelta}{}^2 \geq \dfrac{4}{d-2}.
\end{equation}
This is because all interactions in any string-frame effective action, in terms of the 10-dimensional dilaton $\sdil$, are weighed by string-coupling powers of the form $\smash{f(\sdil) = \e^{- \chi_{\mathrm{E}} \sdil}}$, with $\smash{\chi_{\mathrm{E}}}$ being the Euler number that weighs the perturbative order via the string-worldsheet topology: as the maximum value, for tree-level interactions, corresponds to a sphere $\smash{\chi_{\mathrm{E}} (\mathrm{S}^2) = 2}$, one can never violate eq. (\ref{dilaton-gamma lower bound}). Because $\smash{(\gamma_\infty)^2 \geq \gamma_{\tdelta}{}^2}$, this rules out late-time accelerated expansion in all string-theoretic constructions with positive-definite scalar-potential terms in which the $d$-dimensional dilaton is one of the rolling scalars.

If some of the scalar-potential terms are negative-definite, the bound in eq. (\ref{epsilon bound}), together with the dilaton coupling in eq. (\ref{dilaton-gamma lower bound}), does not automatically give an insurmountable obstruction. It is harder to draw general conclusions because the dilaton could appear satisfying neither of the requirements $\smash{\gamma_\pm^{\tilde{\delta}} \geq \Gamma_\mp^{\tilde{\delta}}}$. An exception is the situation with only two terms in the potential, with opposite signs: since $\smash{\gamma_\pm^{\tilde{\delta}} = \Gamma_\pm^{\tilde{\delta}}}$, one of the two inequalities $\smash{\gamma_\pm^{\tilde{\delta}} \geq \Gamma_\mp^{\tilde{\delta}}}$ is necessarily in place. Therefore, an accelerating universe involving a rolling dilaton minimally requires at least three terms in the potential, not all positive.

Although the dilaton is in principle coupled to all the scalar-potential terms, it could be stabilized. If the dilaton is not a rolling scalar, then we cannot draw fully general conclusions based on eq. (\ref{epsilon bound}) since the other fields, such as radions and complex-structure moduli, are not universally characterized as they depend on the compactification. This can also be said in phenomenological constructions that disregard a possible string-theoretic origin, since the exponential couplings are not necessarily constrained by universal principles. Qualitatively, a general expectation is the following: a large number of scalar-potential terms has a tendency to ease the restrictions since it makes it harder to fall in the condition $\smash{\gamma_\infty^a > 0}$; conversely, a large number of rolling fields tends to obstruct acceleration, since the coefficient $\smash{(\gamma_\infty)^2}$ is additive. Indeed, more scalar-potential terms tend to flatten the total potential, whereas more scalar fields tend to make it steeper.

Importantly, we stress that the bound applies only to quintessence-like proposals in which one assumes that we are currently observing an asymptotic regime of the cosmological evolution. It does not inform us about inflation since the latter should be realized as a transient solution. Finally, we emphasize that the bound highlights the difficulty of satisfying the slow-roll condition at late-time cosmic stages.

\section{Scaling cosmologies} \label{sec: scaling cosmologies}
Cosmological solutions where the scale factor is of power-law form, i.e. scaling solutions, have a special role: if $\smash{(\gamma_\infty)^2 \leq 4 \, (d-1) / (d-2)}$, due to eq. (\ref{epsilon bound}), at sufficiently late times, the scale factor is bounded from below and from above by power-law evolution; if $\smash{(\gamma_\infty)^2 > 4 \, (d-1) / (d-2)}$, scaling solutions are inevitable, with a power $\smash{p = 1/(d-1)}$. More generally, scaling solutions correspond to critical points of the cosmological autonomous system. In fact, one can typically find scaling solutions that are perturbative late-time attractors \cite{Bergshoeff:2003vb, Hartong:2006rt}. For these reasons, although it is hard to prove that scaling solutions always capture the inevitable late-time behavior of the complete solutions, they deserve a detailed analysis.

Scaling solutions can be characterized analytically \cite{Collinucci:2004iw}.
If they exist, for a given initial time $t_\infty$, the generic scalar-field trajectories corresponding to a scaling solution $a(t) = a_\infty (t/t_\infty)^p$, with $p \geq 1/(d-1)$, can be parameterized as
\begin{equation}
    \phi^a_* (t) = \phi^a_\infty + \dfrac{1}{\kappa_d} \, \alpha^a \; \mathrm{ln} \, \dfrac{t}{t_\infty},
\end{equation}
with the additional condition $\smash{\gamma_{ia} \alpha^a = 2}$. Then, given the unit vector $\smash{\theta^a_* = \alpha^a / \sqrt{\alpha^b \alpha_b}}$, which follows the trajectory of the scalar fields over the moduli space, we can show that the normalized directional derivative of the scalar potential is related to the expansion rate as
\begin{equation} \label{directional gamma}
    \gamma_* = - \dfrac{1}{V(\phi_*)} \, \theta^a_* \dfrac{\der V}{\kappa_d \, \der \phi^a_*} (\phi_*) = \dfrac{2}{\sqrt{d-2}} \, \sqrt{\epsilon}.
\end{equation}
This can be proven by exploiting explicitly the analytic properties of scaling solutions. Therefore, the power-law scale-factor evolution is accelerated -- meaning that the condition $\smash{\epsilon < 1}$ holds -- only if the directional scalar-potential coefficient is bounded as $\smash{\gamma_* < 2 / \sqrt{d-2}}$.

A point that should be emphasized is the following: the $\epsilon$-parameter measures the rate of acceleration of the scale factor and it is defined as $\epsilon=-\dot{H}/H^2$. It can be estimated via the gradient of the potential, i.e.
\begin{equation}
    \gamma = \dfrac{\sqrt{\delta^{ab} \, \der_a V \der_b V}}{\kappa_d V},
\end{equation}
for instance under the slow-roll approximation, by which one may compute $\smash{\epsilon = - \dot{H} / H^2 \simeq (d-2) \, \gamma^2 / 4}$. For theories with finite $\smash{\gamma_\infty}$-coefficients, as dictated by eq. (\ref{epsilon bound}), and for scaling scenarios, the slow-roll approximation is generically invalid. For the former, this is obvious as long as $\smash{(\gamma_\infty)^2 \gtrsim 4/(d-2)}$. For the latter, the terms that should be dropped in the slow-roll approximation, despite being numerically smaller by a factor $p (d-1) \geq 1$, decrease over time in the same parametric way as the terms that would be kept. Therefore, the parameter $\gamma$ is not necessarily a meaningful quantity to describe the expansion rate. In particular, the triangle inequality shows that $\smash{\gamma_* \leq \gamma}$ for any directional derivative $\smash{\gamma_*}$ \cite{Andriot:2022brg}. For scaling cosmologies, it however happens that $\smash{\sqrt{\delta^{ab} \, \der_a V \der_b V}/ (\kappa_d V) = \gamma_*}$.

One can always identify a single scalar that serves as a measure of time. Each term in the potential evolves as
\begin{align*}
    V_i \bigl[\phi_*^a(t)\bigr] = \bigl( \Lambda_i \, \e^{- \kappa_d \gamma_{i a} \phi_\infty^a} \bigr) \, \dfrac{t_\infty^2}{t^2}.
\end{align*}
This means that, whatever combination of scalar fields appears in each $V_i$-term, this is such as to provide a $\smash{-2 \, \mathrm{ln} \, (t/t_\infty)}$-behavior. So, for an arbitrary term $\smash{V_{i_0}}$, we can define a canonically-normalized scalar field via the redefinition
\begin{equation} \label{clock-field definition}
    \tilde{\gamma} \ttau = \gamma_{i_0 a} \phi^a,
\end{equation}
where the parameter $\smash{\tilde{\gamma}}$ and the field $\smash{\ttau}$ are defined by the $\mathrm{O}(n)$-rotation for the specific $\smash{\gamma_{i_0 a}}$-coefficients, and where the field evolves over time as
\begin{equation} \label{clock-field time evolution}
    \ttau_*(t) = \ttau_\infty + \dfrac{1}{\kappa_d} \dfrac{2}{\tilde{\gamma}} \; \mathrm{ln} \, \dfrac{t}{t_\infty}.
\end{equation}
Indeed, since the action is independent of $\mathrm{O}(n)$-rotations in the moduli space, this is the only possible behavior such that the potential $\smash{V_{i_0} = \Lambda_{i_0} \, \e^{- \kappa_d \tilde{\gamma} \ttau}}$ can evolve in the right way. In the remaining $m-1$ scalar-potential terms, one generally finds linear combinations of the field $\smash{\ttau}$ and other $n-1$ canonically-normalized scalar fields: the latter evolve with different slopes, but in such a way that all contributions to the scalar potential have a $\smash{(1/t^2)}$-behavior. As the scalar $\smash{\ttau}$ provides a measure of time, all scalar-potential terms evolve with an on-shell behavior that is captured by $\smash{t = t_\infty \, \e^{ \kappa_d \tilde{\gamma} \, (\ttau_* - \ttau_\infty) / 2}}$. For this reason, we can also express the Hubble scale as
\begin{equation} \label{Hubble-scale evolution}
    l_H = \dfrac{1}{H} = l_\infty \, \e^{\frac{1}{2} \, \kappa_d \tilde{\gamma} \ttau_*},
\end{equation}
where $\smash{l_\infty = (t_\infty / p) \, \e^{- \frac{1}{2} \, \kappa_d \tilde{\gamma} \ttau_\infty}}$ is the initial value. It is important to stress that, although one may always define a scalar $\smash{\ttau}$ that appears alone in one scalar-potential term, all scalar-potential terms participate in the cosmological evolution as they fall off over time in the same way and therefore contribute with the same weight to the total energy density.

Another possible field redefinition that one can make is that of a scalar $\smash{\tvarphi}$ that is aligned with the field trajectory in the moduli space. Once the trajectories $\smash{\phi_*^a}$ have been identified, this can be defined by an $\mathrm{O}(n)$-rotation where $\smash{\tvarphi}$ is parallel to the vector $\smash{\theta_*^a}$ and the remaining fields $\smash{\check{\phi}^{\check{a}}}$ are orthogonal to it, with the $\smash{\check{a}}$-index not including $\smash{\tvarphi}$. All the scalar-potential terms can then be written as
\begin{align*}
    V_i = \Lambda_i \, \e^{- \kappa_d \gamma_* \tvarphi - \kappa_d \check{\gamma}_{i \check{a}} \check{\phi}^{\check{a}}},
\end{align*}
where $\smash{\gamma_*}$ is the  directional derivative of eq. (\ref{directional gamma}) and the coefficients $\smash{\check{\gamma}_{i \check{a}}}$ are instead defined by the inverse rotation. By definition, the time-evolution of $\smash{\tvarphi}$ must be
\begin{equation} \label{quintessence-field time evolution}
    \tvarphi_*(t) = \tvarphi_\infty + \dfrac{1}{\kappa_d} \dfrac{2}{\gamma_*} \; \mathrm{ln} \, \dfrac{t}{t_\infty},
\end{equation}
with the other fields being constants $\smash{\check{\phi}_*^{\check{a}} = \check{\phi}_\infty^{\check{a}}}$. This guarantees again the right scalar-potential evolution. In particular, all the fields $\smash{\check{\phi}_*^{\check{a}}}$ can be absorbed into redefinitions of the constants $\Lambda_i$, so as to have a total potential of the form
\begin{align}
    V = \Lambda \, \e^{- \kappa_d \gamma_* \tvarphi_*}.
\end{align}
Here, we stress that neither the constant $\smash{\Lambda}$ nor the coefficient $\smash{\gamma_*}$ can be read off simply from the dimensional reduction of a single dominating term and that this identity is only true on the field-equation solutions. Evidently, one can also use $\smash{\tvarphi}$ to measure time instead of $\smash{\ttau}$.

Consider the 10-dimensional Einstein-frame scalar potential $\smash{\hat{V}_{(10)} = \hat{\Lambda}_{(10)} \, \e^{\hat{\kappa}_{10} \hat{\gamma}_{(10)} \hat{\sdil}_{(10)}}}$, where $\smash{\hat{\sdil}_{(10)}}$ is the canonically-normalized dilaton and $\smash{\hat{\kappa}_{10}}$ defines the gravitational coupling. Upon dimensional reduction, in terms of the canonically-normalized 10-dimensional dilaton $\smash{\tsdil}$ and Einstein-frame radion $\smash{\tomega}$, this corresponds to the $d$-dimensional Einstein-frame potential
\begin{align*}
    V_{(10)} = \Lambda_{(10)} \, \e^{\kappa_d \hat{\gamma}_{(10)} \tsdil - \frac{1}{\sqrt{2}} \frac{\sqrt{10-d}}{\sqrt{d-2}} \, \kappa_d \tomega} = \Lambda_{(10)} \, \e^{- \kappa_d \tilde{\gamma} \ttau},
\end{align*}
where the field $\smash{\ttau}$ has been defined as in eq. (\ref{clock-field definition}). For multi-field models with multi-exponential potentials, one cannot relate $\smash{\tilde{\gamma}}$ to the $\epsilon$-parameter without knowledge of the other scalar-potential terms, as anticipated above. One can only write the on-shell identities
\begin{align*}
    \gamma_* (\tvarphi_* - \tvarphi_\infty) = \gamma_{ia} (\phi_*^a - \phi_\infty^a) = \tilde{\gamma} (\ttau_* - \ttau_\infty),
\end{align*}
but the determination of the constant $\smash{\gamma_*}$ and the definition of the field $\smash{\tvarphi_*}$ cannot disregard the other exponential potentials, if any. If the dilaton is among the rolling scalars contributing non-trivially to the bound in eq. (\ref{epsilon bound}), then eqs. (\ref{dilaton-gamma lower bound}, \ref{directional gamma}) imply the constraint
\begin{equation}
    \gamma_* \geq \dfrac{2}{\sqrt{d-2}}.
\end{equation}
On the other hand, as the value $\smash{\tilde{\gamma} = 2 / \sqrt{d-2}}$ is preserved under dimensional reduction, one is led to conjecture that this should be the minimum possible value \cite{Rudelius:2021oaz, Rudelius:2022gbz}. For a single-term potential, the coefficient thus computed is also the one that controls the acceleration rate, i.e. $\smash{\gamma_* = \tilde{\gamma}}$, and the minimum value argued in terms of dimensional reductions corresponds to the minimum value that is allowed in a string-theoretic effective action. For multi-exponential potentials, this is not the case. We also note that the scalar considered in this argument, i.e. $\smash{\ttau = - \tdelta + (\sqrt{d-2} / \sqrt{10-d}) \, \tsigma}$, corresponds to the combination appearing in curvature-induced potentials.

For scaling solutions, the Hubble scale and the Kaluza-Klein scale are arranged in a specific hierarchy. On the one hand, the Hubble-scale evolution is in eq. (\ref{Hubble-scale evolution}). On the other hand, in an isotropic compactification, the Kaluza-Klein scale evolves as
\begin{equation} \label{KK-scale evolution}
    l_{\mathrm{KK}, d} = \Bigl( \dfrac{4 \pi}{g_s^2} \Bigr)^{\frac{1}{d-2}} \, l_{\p, d} \, \e^{- \frac{\kappa_d \tdelta}{\sqrt{d-2}} + \frac{\kappa_d \tsigma}{\sqrt{10-d}}}.
\end{equation}
In particular, the moduli dependence of the Kaluza-Klein scale is proportional to the moduli dependence of the potential induced by a non-trivial curvature $\smash{\breve{R}_{(10-d)}}$, as computed in terms of a fiducial metric with a string-size fiducial volume. Because this potential, if present, can be used to measure the Hubble scale, we can write $\smash{2 \tdelta / \sqrt{d-2} - 2 \tsigma / \sqrt{10-d} = - \tilde{\gamma} \ttau}$ and thus find
\begin{equation} \label{no scale separation}
    \dfrac{l_H^2}{l_{\mathrm{KK}, d}^2} = \dfrac{\bigl[ (d-1) \, p - 1 \bigr] \, (d-2)}{p \xi \, (- l_s^2 \breve{R}_{(10-d)})}.
\end{equation}
Here, $\xi = V / V_R$ is the order-1 ratio between the on-shell total potential $V$ and the curvature-induced potential $V_R$. This means that the theory is by no means genuinely $d$-dimensional. If the internal curvature is trivial, the details of the dilaton and radion time evolution are necessary to assess the ratio.

\section{Discussion} \label{sec: discussion}

In this letter we reveal a very simple but generic bound for the $\epsilon$-parameter of any cosmological model where the potential is a sum of exponentials. This is a restriction that emerges at late times since the solutions to the cosmological equations exhibit universal asymptotic features. Due to this bound, we show that the idea of our universe approaching an asymptotic region of the field space at late time \cite{Conlon:2022pnx, Rudelius:2022gbz, Calderon-Infante:2022nxb} is difficult to realize. Solutions with negatively-curved spatial slices \cite{Marconnet:2022fmx} are not constrained by our bound.

The bound that we derived provides us with a necessary condition to quickly diagnose whether cosmic acceleration is even possible at all, with a much wider applicability than previous studies \cite{Rudelius:2021azq, Rudelius:2022gbz, Calderon-Infante:2022nxb}; we also stress that, even if the bound does not rule out cosmic acceleration, one then needs to identify whether the expansion is actually accelerated. The fact that the presence of negative-definite terms in the potential relaxes the bounds on acceleration resonates with other well-known no-go results \cite{Gibbons:1984kp, deWit:1986mwo, Maldacena:2000mw}, but it sharpens the specific situation of multi-field multi-exponential potentials as it provides a quantitative restriction. Another result is that the slow-roll approximation can also be ruled out by our no-go result. This conclusion is in line with previous findings in the string-theory literature \cite{Hertzberg:2007wc, Garg:2018reu, Garg:2018zdg, Cicoli:2021fsd, Andriot:2022xjh, Hebecker:2019csg, ValeixoBento:2020ujr}, although such claims are based on the gradient and the curvature of the potential, which, as we have discussed, are not the right criteria for cosmic acceleration in general. On the other hand, transient solutions -- viable for inflation -- are not within the reach of the current analysis.

In particular, our bound represents an important no-go result that shows a generic difficulty for string-theoretic models to entail accelerated expansion at late times. To be at weak string coupling, then either the dilaton has to be stabilized or there need to be more than two terms in the potential, with at least a positive- and a negative-definite one. If the dilaton is not stabilized, acceleration is generally favored in theories in which many terms of both signs appear in the scalar potential -- which can be the case, for instance, for potentials generated by string-loop corrections in non-supersymmetric theories.

For scaling solutions, we also show that the $\epsilon$-parameter is in an exact correspondence with the normalized directional derivative of the scalar potential along the field-space trajectories. Noticeably, curvature-induced scalar potentials stand out as special, when a scaling solution is possible: they correspond to minimal tree-level potentials that provide the dilaton bound, and, whenever they are present, they prevent a parametric separation of the Kaluza-Klein and Hubble scales.

Intriguing future directions of research regard the fate of scaling solutions as the inevitable attractors of the cosmological equations. Although there are several hints that this may be the case, a general proof is unknown \cite{Shiu:2023rxt}. This would substantiate our claims on scaling solutions not just as possible solutions, but as the inevitable late-time solutions. A study of the attractor mechanisms towards scaling solutions along the lines of our analysis can also make it possible to quantify the duration of transient accelerated phases identified as early inflationary epochs. A conjectured string-theoretic feature is the possible presence of a large number of unstabilized complex-structure moduli \cite{Bena:2020xrh, Grana:2022dfw}: although it may be tempting to consider such fields among the rolling scalars at the boundary of moduli space, they are expected to have a tendency to raise the value of the constant $\smash{(\gamma_\infty})^2$, which is additive, thus contributing to deceleration. Recent studies of the flux potential in the asymptotic region of the complex structure moduli space \cite{Grimm:2019ixq} may allow us to confirm this expectation. Finally, our findings point to the importance of dynamical (rather than just kinematical) considerations in the Swampland Program \cite{Landete:2018kqf, Bedroya:2019snp, Apers:2022cyl}.

\vspace{6pt}

\textbf{Note added.} Shortly after this work appeared as a preprint, we continued the program of analytical characterization of late-time cosmological solutions in ref. \cite{Shiu:2023rxt}. Here, we defined a convex-hull formulation of our bound and we proved sufficient conditions for scaling cosmologies to be inevitable late-time attractors.

\onecolumngrid
\newpage

\appendix

\section{Late-time bounds on cosmological autonomous systems} \label{app: late-time bounds on cosmological autonomous systems}

Let $\smash{V = \sum_{i = 1}^m \Lambda_i \, \e^{- \kappa_d \gamma_{i a} \phi^a}}$ be the scalar potential for the canonically-normalized scalar fields $\phi^a$, with $a=1,\dots,n$, in a $d$-dimensional FLRW-background: one can reduce the cosmological scalar-field and Friedmann equations to a system of autonomous equations \cite{Halliwell:1986ja, Copeland:1997et, Coley:1999mj, Guo:2003eu}. In terms of the variables
\begin{align*}
    x^a & = \dfrac{\kappa_d}{\sqrt{d-1} \sqrt{d-2}} \, \dfrac{\dot{\phi}^a}{H}, \\
    y^i & = \dfrac{\kappa_d \sqrt{2}}{\sqrt{d-1} \sqrt{d-2}} \, \dfrac{1}{H} \, \sqrt{\Lambda_i \, \e^{- \kappa_d \gamma_{i a} \phi^a}},
\end{align*}
and defining for simplicity
\begin{align*}
    f & = (d-1) H, \\
    c_{i a} & = \dfrac{1}{2} 
    \dfrac{\sqrt{d-2}}{\sqrt{d-1}} \, \gamma_{ia},
\end{align*}
the cosmological equations can indeed be expressed as
\begin{subequations}
\begin{align}
    \dot{x}^a & = \biggl[ -x^a (y)^2 + \sum_{i=1}^m {c_i}^{a} (y^i)^2 \biggr] \, f, \label{x-equation} \\
    \dot{y}^i & = \bigl[ (x)^2 - c_{ia} x^a \bigr] \, y^i f, \label{y-equation}
\end{align}
\end{subequations}
jointly with the two conditions
\begin{subequations}
\begin{align}
    & \dfrac{\dot{f}}{f^2} = - (x)^2, \label{f-equation} \\
    & (x)^2 + (y)^2 = 1. \label{sphere-condition}
\end{align}
\end{subequations}
(Here, the position of the $i$- and $a$-indices is arbitrary, since the former are just dummy labels and the latter refer to a field-space metric that is a Kronecker delta; if they appear together, such as in the coefficients $\smash{c_{ia}}$ and $\smash{{c_i}^a}$, we always write them on the left and on the right, respectively, to avoid confusion, and we place the $a$-index in the same (upper or lower) position as the other $x^a$-types variable that appear in the expression of reference. Einstein summation convention is understood for the metric-contraction on the $a$-indices and moreover we use the shorthand notations $\smash{(x)^2 = x_a x^a}$ and $\smash{(y)^2 = \sum_{i=1}^m (y^i)^2}$.)

Below are a series of mathematical results. A formulation in terms of physical observables and an analysis of their implications is in the main text. Perturbative analyses of the late-time behavior of autonomous systems in cosmological scenarios appear in refs. \cite{Malik:1998gy, Coley:1999mj, Guo:2003eu, Kim:2005ne, Hartong:2006rt}. \\

Let the unknown functions be such that $x^a \in [-1, 1]$ and $y^i \in [-1,1]$ and let $c^a$ denote the minimum of the constant parameters for each $a$-index, i.e. $\smash{c^a = \min_i {c_i}^a}$. For each $a$-index, if $c^a > 0$, let $\smash{c_\infty^a = c^a}$; if $c^a \leq 0$, then let $\smash{c_\infty^a = 0}$. Let $t_0$ be the initial time.

\vspace{8pt}
\begin{lemma} \label{lemma: f - lower bound}
    For $t > t_0$, one has
    \begin{equation} \label{eq.: f - lower bound}
        f(t) \geq \dfrac{1}{t-t_0 + \dfrac{1}{f(t_0)}}.
    \end{equation}
    This also implies that $\smash{\int_{t_0}^\infty \de t \, f(t) = \infty}$.
\end{lemma}

\begin{proof}
Let $\smash{f_-}$ be a function such that $\smash{\dot{f}_- / (f_-)^2 = - 1}$. A simple integration gives $\smash{f_- = 1 / [t - t_0 + 1/f_-(t_0)]}$. Because $\smash{(x)^2 \leq 1}$, eq. (\ref{f-equation}) gives $\smash{\dot{f}/f^2 \geq \dot{f}_- / (f_-)^2}$; by integrating the condition $\smash{\de f / f^2 \geq \de f_- / (f_-)^2}$, one immediately gets eq. (\ref{eq.: f - lower bound}).
\end{proof}

\vspace{8pt}
\begin{lemma} \label{lemma: x - lower bound}
    If $\smash{c^a \leq 1}$ and if $\smash{\int_{t_0}^\infty \de s \; f(s) \bigl[ y(s) \bigr]^2 = \infty}$, then one has
    \begin{equation} \label{eq.: x - lower bound}
        \liminf_{t \to \infty} x^a(t) \geq c^a.
    \end{equation}
\end{lemma}

\begin{proof}
    In view of eq. (\ref{x-equation}), one can write the inequality
    \begin{align*}
        \dot{x}^a = \biggl[ -x^a (y)^2 + \sum_{i} {c_{i}}^{a} (y^i)^2 \biggr] \, f \geq \bigl[ -x^a + c^a \bigr] \, f (y)^2.
    \end{align*}
    Then, defining the functions $\smash{\xi(t) = f(t) \bigl[ y(t) \bigr]^2}$ and $\smash{\varphi(t) = \int_{t_0}^t \de s \; \xi(s)}$, one can write
    \begin{align*}
        \dfrac{\de}{\de t} \, \bigl[ \e^{\varphi(t)} x^a(t) \bigr] = \e^{\varphi(t)} \bigl[ \dot{x}^a(t) + x^a(t) \, \xi(t) \bigr] \geq c^a \, \e^{\varphi(t)} \xi(t).
    \end{align*}
    By integrating this inequality, one then finds the further inequality
    \begin{align*}
        \e^{\varphi(t)} x^a(t) - x^a(t_0) \geq  c^a \int_{t_0}^t \de s \; \e^{\varphi(s)} \xi(s) = c^a \, \bigl[ \e^{\varphi(t)} - 1 \bigr].
    \end{align*}
    If $\smash{c^a \leq 1}$ and if $\smash{\lim_{t \to \infty} \varphi(t) = \infty}$, then the conclusion in eq. (\ref{eq.: x - lower bound}) holds.
\end{proof}

\vspace{8pt}
\begin{corollary}
    If $c^a > 1$ for at least one $a$-index, then $\smash{\int_{t_0}^\infty \de s \; f(s) \bigl[ y(s) \bigr]^2 < \infty}$.
\end{corollary}

\begin{proof}
    If $c^a > 1$, then eq. (\ref{eq.: x - lower bound}) cannot hold by eq. (\ref{sphere-condition}), which forces the inequality $(x)^2 \leq 1$ and therefore the inequality $(x^a)^2 \leq 1$. Therefore, one cannot have $\smash{\lim_{t \to \infty} \varphi(t) = \infty}$.
\end{proof}

\vspace{8pt}
\begin{corollary} \label{corollary: f - upper bound}
    If $(c_\infty)^2 \leq 1$, and if $\smash{\int_{t_0}^\infty \de s \; f(s) \bigl[ y(s) \bigr]^2 = \infty}$, for $t > t_\infty$, where $t_\infty$ is a time $t_\infty > t_0$, one has
    \begin{equation} \label{eq.: f - upper bound}
        f(t) \leq \dfrac{1}{[c(\varphi_\infty)]^2 \, (t-t_\infty) + \dfrac{1}{f(t_\infty)}}.
    \end{equation}
    Here, $\smash{c(\varphi_\infty)}$ is a quantity such that $\smash{\lim_{t_\infty \to \infty} c(\varphi_\infty) = c_\infty}$.
\end{corollary}

\begin{proof}
    If $(c_\infty)^2 \leq 1$, then $\smash{c^a \leq 1}$ for all $a$-indices. For an arbitrarily large number $\varphi_\infty$, there exists a time $t_\infty$ such that $\varphi(t) > \varphi_\infty$ for $t > t_\infty$. Then, because one can write
    \begin{align*}
        x^a(t) \geq c^a + \e^{-\varphi(t)} \bigl[ x^a(t_0) - c^a \bigr] \geq c^a - \e^{-\varphi(t)} \bigl\ab x^a(t_0) - c^a \bigr\ab,
    \end{align*}
    for $t > t_\infty$, it is also possible to write the inequality
    \begin{align*}
        x^a(t) \geq c^a - \e^{-\varphi(t)} \bigl\ab x^a(t_0) - c^a \bigr\ab \geq c^a - \e^{-\varphi_\infty} \bigl\ab x^a(t_0) - c^a \bigr\ab.
    \end{align*}
    If $\smash{c^a = c_\infty^a > 0}$, let $\smash{c(\varphi_\infty)^a = c^a - \e^{-\varphi_\infty} \bigl\ab x^a(t_0) - c^a \bigr\ab}$; if $\smash{c^a \leq 0}$, let $\smash{c(\varphi_\infty)^a = c_\infty^a = 0}$. Thus, if $(c_\infty)^2 \leq 1$, which implies $[c(\varphi_\infty)]^2 \leq (c_\infty)^2 \leq 1$, for $t > t_\infty$, the inequality holds $(x)^2 \geq [c(\varphi_\infty)]^2$. Then, let $\smash{f_\infty}$ be a function such that $\smash{\dot{f}_\infty / f_\infty^2 = - [c(\varphi_\infty)]^2}$. A simple integration gives $\smash{f_\infty = 1 / [[c(\varphi_\infty)]^2 (t - t_\infty) + 1/f_\infty(t_\infty)]}$. Because $\smash{- (x)^2 \leq - [c(\varphi_\infty)]^2}$ for $t > t_\infty$, eq. (\ref{f-equation}) gives $\smash{\dot{f}/f^2 \leq \dot{f}_\infty / f_\infty^2}$; by integrating the condition $\smash{\de f / f^2 \leq \de f_\infty / f_\infty^2}$, one immediately gets eq. (\ref{eq.: f - upper bound}).
\end{proof}

\vspace{8pt}
\begin{corollary}
    If $(c_\infty)^2 > 1$, then $\smash{\int_{t_0}^\infty \de s \; f(s) \bigl[ y(s) \bigr]^2 < \infty}$.
\end{corollary}

\begin{proof}
    By contradiction, let $\smash{\lim_{t \to \infty} \varphi(t) = \infty}$. Then one can write $\smash{(x)^2 \geq (c_\infty)^2 > 1}$, which is impossible by eq. (\ref{sphere-condition}). Therefore, one cannot have $\smash{\lim_{t \to \infty} \varphi(t) = \infty}$.
\end{proof}

\vspace{8pt}
\begin{lemma} \label{lemma: non-proper attractors}
    If $\smash{\int_{t_0}^\infty \de s \; f(s) \bigl[ y(s) \bigr]^2 < \infty}$, then
    \begin{subequations}
    \begin{align}
        \lim_{t \to \infty} x^a(t) & = \tilde{x}^a,  \label{eq.: unique non-proper x-attractor} \\
        \lim_{t \to \infty} y^i(t) & = 0.  \label{eq.: unique non-proper y-attractor}
    \end{align}
    \end{subequations}
    Note in particular that $\smash{(\tilde{x})^2 = 1}$.
\end{lemma}

\begin{proof}
    For any given $a$-index, if there exists a positive $c_{i a} > 0$, let $\smash{b_a = \max_i c_{ia} > 0}$, else let $\smash{b_a = 0 \geq \max_i c_{ia}}$. Then, the inequalities hold
    \begin{align*}
        \bigl[ -x^a + b^a \bigr] \, f (y)^2 \geq \dot{x}^a \geq \bigl[ -x^a + c^a \bigr] \, f (y)^2.
    \end{align*}
    If $\smash{\dot{x}^a \geq 0}$, then
    \begin{align*}
        \bigl\ab \dot{x}^a \bigr\ab \leq \bigl\ab -x^a + b^a \bigr\ab \, f (y)^2 \leq \bigl[ 1 + b^a \bigr] \, f (y)^2.
    \end{align*}
    If $\smash{\dot{x}^a < 0}$, then
    \begin{align*}
        \bigl\ab \dot{x}^a \bigr\ab \leq \bigl\ab -x^a + c^a \bigr\ab \, f (y)^2 \leq \bigl[ 1 + \ab c^a \ab \bigr] \, f (y)^2.
    \end{align*}
    Therefore, defining $\smash{\tilde{c}^a = \max \lbrace b^a, \ab c^a \ab \rbrace}$ one can write
    \begin{align*}
        \bigl\ab \dot{x}^a \bigr\ab \leq \bigl[ 1 + \tilde{c}^a \bigr] \, f (y)^2.
    \end{align*}
    Because $\smash{\int_{t_0}^\infty \de s \; f(s) \bigl[ y(s) \bigr]^2 < \infty}$, for any arbitrary constant $\epsilon^a > 0$, there exists a time $\smash{t_{\epsilon^a} > t_0}$ such that
    \begin{align*}
        \int_{t_{\epsilon^a}}^\infty \de s \; f(s) \bigl[ y(s) \bigr]^2 < \dfrac{\epsilon^a}{1 + \tilde{c}^a},
    \end{align*}
    therefore, for a time $t_*$ in the range $\smash{t > t_* > t_{\epsilon^a}}$, one can write
    \begin{align*}
        \bigl\ab x^a(t) - x^a(t_*) \bigr\ab = \biggl\ab \int_{t_*}^t \de s \; \dot{x}^a(s) \biggr\ab \leq \int_{t_*}^t \de s \; \bigl\ab \dot{x}^a(s) \bigr\ab \leq \bigl[ 1 + \tilde{c}^a \bigr] \int_{t_*}^t \de s \; f(s) \bigl[ y(s) \bigr]^2 < \epsilon^a.
    \end{align*}
    This means that $\smash{\lbrace x^a(t) \rbrace_{t \geq t_0}}$ is a Cauchy sequence, which implies the existence of the limit
    \begin{align*}
        \lim_{t \to \infty} x^a(t) = \tilde{x}^a.
    \end{align*}
    Furthermore, this implies the existence of the limit $\smash{\lim_{t \to \infty} \bigl[x(t)\bigr]^2 = (\tilde{x})^2}$. Therefore, for any arbitrary constant $\delta > 0$, there exists a time $\smash{t_{\delta} > t_0}$ such that $\smash{\bigl\ab \bigl[x(t)\bigr]^2 - (\tilde{x})^2 \bigr\ab < \delta}$ at all times $\smash{t > t_\delta}$. For all times $t > t_\delta$, one can thus write the inequality
    \begin{align*}
        \bigl[ x(t) \bigr]^2 < (\tilde{x})^2 + \delta.
    \end{align*}
    Assuming the condition $\smash{(\tilde{x})^2 < 1}$ to hold by contradiction, then one can fix the $\delta$-constant as $\smash{\delta = \bigl[ 1 - (\tilde{x})^2 \bigr]/2 > 0}$, which implies the inequality
    \begin{align*}
        \bigl[ x(t) \bigr]^2 < \dfrac{1 + (\tilde{x})^2}{2}.
    \end{align*}
    Therefore, one can write
    \begin{align*}
        \int_{t_\delta}^\infty \de s \; f(s) \bigl[ 1 - \bigl[x(s)\bigr]^2 \bigr] > \int_{t_\delta}^\infty \de s \; f(s) \bigl[ 1 - \dfrac{1 + (\tilde{x})^2}{2} \bigr] = \dfrac{\bigl[1 - (\tilde{x})^2 \bigr]}{2} \int_{t_\delta}^\infty \de s \; f(s) = \infty,
    \end{align*}
    due to eq. (\ref{eq.: f - lower bound}), which is inconsistent with the hypothesis. Therefore, one must have $\smash{(\tilde{x})^2 = 1}$.
\end{proof}

\vspace{12pt}

One can constrain $x^a$ further than in eq. (\ref{eq.: x - lower bound}) and consequently tighten the bound in eq. (\ref{eq.: f - upper bound}). Let $\smash{C_a = \max_i c_{i a}}$.

\vspace{8pt}
\begin{lemma} \label{lemma: x - upper bound}
    If $C^a \geq - 1$ and if $\smash{\int_{t_0}^\infty \de s \; f(s) \bigl[ y(s) \bigr]^2 = \infty}$, then one has
    \begin{equation} \label{eq.: x - upper bound}
        \limsup_{t \to \infty} x^a(t) \leq C^a.
    \end{equation}
\end{lemma}

\begin{proof}
    In view of eq. (\ref{x-equation}), one can write the inequality
    \begin{align*}
        \dot{x}^a = \biggl[ -x^a (y)^2 + \sum_{i} {c_{i}}^{a} (y^i)^2 \biggr] \, f \leq \bigl[ -x^a + C^a \bigr] \, f (y)^2.
    \end{align*}
    Then, defining the functions $\smash{\xi(t) = f(t) \bigl[ y(t) \bigr]^2}$ and $\smash{\varphi(t) = \int_{t_0}^t \de s \; \xi(s)}$, one can write
    \begin{align*}
        \dfrac{\de}{\de t} \, \bigl[ \e^{\varphi(t)} x^a(t) \bigr] = \e^{\varphi(t)} \bigl[ \dot{x}^a(t) + x^a(t) \, \xi(t) \bigr] \leq C^a \e^{\varphi(t)} \xi(t).
    \end{align*}
    By integrating this inequality, one then finds the further inequality
    \begin{align*}
        \e^{\varphi(t)} x^a(t) - x^a(t_0) \leq C^a \int_{t_0}^t \de s \; \e^{\varphi(s)} \xi(s) = C^a \, \bigl[ \e^{\varphi(t)} - 1 \bigr].
    \end{align*}
    If $C^a \geq - 1$ and if $\smash{\lim_{t \to \infty} \varphi(t) = \infty}$, then the conclusion in eq. (\ref{eq.: x - upper bound}) holds.
\end{proof}

\vspace{8pt}
\begin{corollary}
    If $C^a < -1$ for at least one $a$-index, then $\smash{\int_{t_0}^\infty \de s \; f(s) \bigl[ y(s) \bigr]^2 < \infty}$.
\end{corollary}

\begin{proof}
    If $C^a < - 1$, then eq. (\ref{eq.: x - upper bound}) cannot hold by eq. (\ref{sphere-condition}), which forces the inequality $(x)^2 \leq 1$ and therefore the inequality $(x^a)^2 \leq 1$. Therefore, one cannot have $\smash{\lim_{t \to \infty} \varphi(t) = \infty}$.
\end{proof}

\vspace{8pt}
\begin{remark}
    In view of eqs. (\ref{eq.: x - lower bound}, \ref{eq.: x - upper bound}), one can improve the constraint in eq. (\ref{eq.: f - upper bound}). Let $\smash{c^a \leq 1}$ and $\smash{C^a \geq -1}$. For each $a$-index, if $\smash{c^a > 0}$, then $\smash{\liminf_{t \to \infty} (x^a)^2 \geq (c^a)^2}$, while if $\smash{C^a < 0}$, then $\smash{\liminf_{t \to \infty} (x^a)^2 \geq (C^a)^2}$. This can be used to reformulate the definition of $\smash{c_\infty^a}$ as
    \begin{align*}
        c_\infty^a = \left\lbrace\!
        \begin{array}{lcl}
            c^a, & & c^a > 0, \; C^a > 0, \\
            0, & & c^a < 0, \; C^a > 0, \\
            \ab C^a \ab, & & c^a < 0, \; C^a < 0.
        \end{array}
        \right.
    \end{align*}
    In view of constraining $f$, it is interesting to find the bounds that keep $\smash{(x)^2}$ as far from the origin as possible.
\end{remark}

\vspace{8pt}
\begin{remark}
    In view of eqs. (\ref{x-equation}, \ref{y-equation}), by defining $\smash{\hat{x}^a = - x^a}$, one can reformulate the problem in terms of the set of constants $\smash{{\hat{c}_i}\vphantom{c_i}^a = - {c_i}^a}$. Because $\smash{\sup x^a = - \inf \hat{x}^a}$ and $\smash{\hat{c}^a = \min_i {\hat{c}_i}\vphantom{c_i}^a = - \max_i {c_i}^a = - C^a}$, the constraint $\smash{\limsup_{t \to \infty} x^a(t) \leq C^a}$ in eq. (\ref{eq.: x - upper bound}) is equivalent to the constraint $\smash{\liminf_{t \to \infty} \hat{x}^a(t) \geq \hat{c}^a}$ in eq. (\ref{eq.: x - lower bound}). Therefore, one can always redefine the variables in such a way that the constraint on the function $f$ in eq. (\ref{eq.: f - upper bound}) holds in full generality.
\end{remark}

\vspace{12pt}

Part of the above results can be generalized to the case in which some of the $y^i$-functions are imaginary. Let $\smash{(y^{i_+})^2 = (y_+^{i_+})^2 > 0}$ and $\smash{(y^{i_-})^2 = - (y_-^{i_-})^2 > 0}$, distinguishing over all possible indices $\smash{i = i_+, i_-}$. For each $a$-index, let $\smash{c_+^a = \min_{i_+} {c_{i_+}}^a}$ and $\smash{C_-^a = \max_{i_-} {c_{i_-}}^a}$, and let $\smash{C_+^a = \max_{i_+} {c_{i_+}}^a}$ and $\smash{c_-^a = \min_{i_-} {c_{i_-}}^a}$. As a further assumption, let $\smash{(y_+)^2 > (y_-)^2}$ at all times.

\vspace{8pt}
\begin{lemma} \label{lemma: x - lower bound 2}
    If $\smash{c_+^a \leq 1}$, with $\smash{c_+^a \geq C_-^a}$, and if $\smash{\int_{t_0}^\infty \de s \; f(s) \, \bigl\lbrace \bigl[ y_+(s) \bigr]^2 - \bigl[ y_-(s) \bigr]^2 \bigr\rbrace = \infty}$, then one has
    \begin{equation} \label{eq.: x - lower bound 2}
        \liminf_{t \to \infty} x^a(t) \geq c_+^a.
    \end{equation}
\end{lemma}

\begin{proof}
    In view of eq. (\ref{x-equation}), one can write the inequality
    \begin{align*}
        \dot{x}^a & = \biggl[ -x^a \bigl[ (y_+)^2 - (y_-)^2 \bigr] + \sum_{i_+} {c_{i_+}}^{a} (y^{i_+})^2 - \sum_{i_-} {c_{i_-}}^{a} (y^{i_-})^2 \biggr] \, f \\
        & \geq \biggl[ -x^a \bigl[ (y_+)^2 - (y_-)^2 \bigr] + c_+^a (y_+)^2 - C_-^a (y_-)^2 \biggr] \, f.
    \end{align*}
    Because $\smash{c_+^a \geq C_-^a}$, one can further write
    \begin{align*}
        \dot{x}^a \geq \bigl[ -x^a + c_+^a \bigr] \, \bigl[ (y_+)^2 - (y_-)^2 \bigr] \, f
    \end{align*}
    and then replicate exactly the procedure that leads to eq. (\ref{eq.: x - lower bound}). If it were $\smash{c_+^a < C_-^a}$, one would not be able to write useful inequalities.
\end{proof}

\vspace{8pt}
\begin{corollary}
    If $\smash{c_+^a > 1}$ for at least one $a$-index, then $\smash{\int_{t_0}^\infty \de s \; f(s) \, \bigl\lbrace \bigl[ y_+(s) \bigr]^2 - \bigl[ y_-(s) \bigr]^2 \bigr\rbrace < \infty}$.
\end{corollary}

\begin{proof}
    If $\smash{c_+^a > 1}$, then eq. (\ref{eq.: x - lower bound 2}) cannot hold by eq. (\ref{sphere-condition}).
\end{proof}

\vspace{8pt}
\begin{lemma} \label{lemma: x - upper bound 2}
    If $\smash{c_-^a \geq -1}$, with $\smash{c_-^a \geq C_+^a}$, and if $\smash{\int_{t_0}^\infty \de s \; f(s) \, \bigl\lbrace \bigl[ y_+(s) \bigr]^2 - \bigl[ y_-(s) \bigr]^2 \bigr\rbrace = \infty}$, then one has
    \begin{equation} \label{eq.: x - upper bound 2}
        \limsup_{t \to \infty} x^a(t) \leq c_-^a.
    \end{equation}
\end{lemma}

\begin{proof}
    In view of eq. (\ref{x-equation}), one can write the inequality
    \begin{align*}
        \dot{x}^a & = \biggl[ -x^a \bigl[ (y_+)^2 - (y_-)^2 \bigr] + \sum_{i_+} {c_{i_+}}^{a} (y^{i_+})^2 - \sum_{i_-} {c_{i_-}}^{a} (y^{i_-})^2 \biggr] \, f \\
        & \leq \biggl[ -x^a \bigl[ (y_+)^2 - (y_-)^2 \bigr] + C_+^a (y_+)^2 - c_-^a (y_-)^2 \biggr] \, f.
    \end{align*}
    Because $\smash{c_-^a \geq C_+^a}$, one can further write
    \begin{align*}
        \dot{x}^a \leq \bigl[ -x^a + c_-^a \bigr] \, \bigl[ (y_+)^2 - (y_-)^2 \bigr] \, f
    \end{align*}
    and then replicate exactly the procedure that leads to eq. (\ref{eq.: x - upper bound}). Therefore, If $\smash{c_-^a \geq - 1}$ and if $\smash{\lim_{t \to \infty} \varphi(t) = \infty}$, then the conclusion in eq. (\ref{eq.: x - upper bound 2}) holds. If it were $\smash{c_-^a < C_+^a}$, one would not be able to write useful inequalities.
\end{proof}

\vspace{8pt}
\begin{corollary}
    If $\smash{c_-^a < -1}$ for at least one $a$-index, then $\smash{\int_{t_0}^\infty \de s \; f(s) \, \bigl\lbrace \bigl[ y_+(s) \bigr]^2 - \bigl[ y_-(s) \bigr]^2 \bigr\rbrace < \infty}$.
\end{corollary}

\begin{proof}
    If $\smash{c_-^a < 1}$, then eq. (\ref{eq.: x - upper bound 2}) cannot hold by eq. (\ref{sphere-condition}).
\end{proof}

\vspace{8pt}
\begin{remark}
    For an arbitrary $\smash{{c_i}^a}$-matrix, with arbitrary signs of the $\smash{(y^i)^2}$-variables and the condition that $\smash{(y)^2 > 0}$ at all times $t > t_0$, one can find valuable constraints for $\smash{x^a}$ if at least one of the inequalities $\smash{c_+^a \geq C_-^a}$ and $\smash{c_-^a \geq C_+^a}$ is satisfied. One can have either one or none of them to hold: in the latter case no useful bound can be written.
\end{remark}

\vspace{8pt}
\begin{remark} \label{remark: bound hierarchy}
    In view of eqs. (\ref{x-equation}, \ref{y-equation}), if $\smash{c_-^a \geq C_+^a}$, one can define the variables $\smash{\hat{x}^a = - x^a}$, so as to reformulate the problem in terms of the coefficients $\smash{{\hat{c}_i}\vphantom{c_i}^a = - {c_i}^a}$. Because $\smash{\sup x^a = - \inf \hat{x}^a}$, and because $\smash{\hat{c}_+^a = \min_{i_+} {\hat{c}_{i_+}}\vphantom{c_i}^a = - \max_{i_+} {c_{i_+}}^a = - C_+^a}$ and $\smash{\hat{C}_-^a = \max_{i_-} {\hat{c}_{i_-}}\vphantom{c_i}^a = - \min_{i_-} {c_{i_-}}^a = - c_-^a}$, this means that the constraint $\smash{\sup x^a(t) \leq c_-^a}$ in eq. (\ref{eq.: x - upper bound 2}) can be expressed as the constraint $\smash{\inf \hat{x}^a(t) \geq \hat{C}_-^a}$: as $\smash{\hat{c}_+^a \geq \hat{C}_-^a}$, this constraint is weaker than the one in eq. (\ref{eq.: x - lower bound 2}). One can reach the same conclusion by considering the case in which $\smash{c_+^a \geq C_-^a}$. Therefore, if either $\smash{c_-^a \geq C_+^a}$ or $\smash{c_+^a \geq C_-^a}$, one can always formulate the problem in terms of the constraint in eq. (\ref{eq.: x - lower bound 2}).
\end{remark}

\vspace{8pt}
\begin{remark} \label{remark: general f - upper bound}
    For each $a$-index, if $\smash{c_+^a > 0}$, let $\smash{c_\infty^a = c_+^a}$; if $\smash{c_+^a \leq 0}$, then let $\smash{c_\infty^a = 0}$. Let $\smash{c_+^a \geq C_-^a}$, with $\smash{(c_\infty)^2 \leq 1}$: one then has the constraint in eq. (\ref{eq.: f - upper bound}).
\end{remark}

\vspace{8pt}
\begin{remark} \label{remark: general non-proper attractors}
    It is apparent that, if $\smash{(y)^2 = (y_+)^2 - (y_-)^2 > 0}$, then eqs. (\ref{eq.: unique non-proper x-attractor}, \ref{eq.: unique non-proper y-attractor}) also hold: in particular, they hold if $\smash{(c_\infty)^2 > 1}$.
\end{remark}

\newpage
\begin{acknowledgments}
\subsection*{Acknowledgments}
GS and FT are supported in part by the DOE grant DE-SC0017647. HT is supported in part by the NSF CAREER grant DMS-1843320 and a Vilas Faculty Early-Career Investigator Award.
\end{acknowledgments}

\bibliographystyle{apsrev4-1}
\bibliography{report.bib}

\end{document}